\documentclass{article}
\usepackage{epsfig}
\usepackage{amsbsy,latexsym}
\usepackage{amsmath}
\usepackage{amssymb, mathrsfs}
\usepackage[mathscr]{eucal}
\usepackage{hyperref}
\usepackage{amsfonts}
\usepackage{amsthm}
\usepackage{amsmath}
\usepackage{amsfonts}
\usepackage{latexsym}
\usepackage{amssymb}
\usepackage{amscd}
\usepackage[latin1]{inputenc}
\usepackage{verbatim}

\usepackage[english]{babel}

\newtheorem{definition}{Definition}[section]
\newtheorem{theorem}[definition]{Theorem}
\newtheorem{lemma}[definition]{Lemma}
\newtheorem{corollary}[definition]{Corollary}
\newtheorem{proposition}[definition]{Proposition}
\theoremstyle{definition}
\newtheorem{remark}[definition]{Remark}

\newtheorem{example}[definition]{Example}

\newcommand\tr{ \operatorname{Tr} } 

\newcommand{\Ep}{\mathcal{E}}
\newcommand{\Md}{\mathcal{M}_{d}}

\newcommand{\Bh}{\mathcal{B}(\mathcal{H})}

\newcommand{\Me}{\mathcal{M}_{\mathcal{E}}}
\newcommand{\Mi}{\mathcal{M}_{\mathcal{E}^{\infty}}}

\title{A new bound on quantum Wielandt inequality}

\author{Mizanur Rahaman}

\begin{document}
\maketitle
{\em
Institute for Quantum Computing and Department of Pure Mathematics,
University of Waterloo,}
{\em Waterloo, Ontario N2L 3G1, Canada} 


\begin{abstract}
A new bound on quantum version of Wielandt inequality for positive (not necessarily completely positive) maps has been established. Also  bounds for entanglement breaking and PPT channels are put forward which are better bound than the previous bounds known. We prove that a primitive positive map $\Ep$ acting on $\Md$ that satisfies the Schwarz inequality becomes strictly positive after at most $2(d-1)^2$ iterations. This is to say, that after $2(d-1)^2$
iterations, such a map sends every positive semidefinite matrix to a positive definite one.
This finding does not depend on the number of Kraus operators as the map may not admit any Kraus decomposition.   The motivation of this work is to provide an answer to a question raised in 
the article \cite{Wielandt} by Sanz-Garc\'ia-Wolf and Cirac.  
\end{abstract}
\section{Introduction}
A $d\times d$ stochastic matrix $W$ is called primitive if there exists a number $k\in \mathbb{N}$ such that $(W^k)_{i,j}>0$ for all $(i,j)$, that is all the entries of $W^k$ is strictly positive. The minimum $k$ for which this occurs, denoted by $p(W)$, is called the (classical) index of primitivity of $W$. The  Wielandt's inequality (\cite{Wielandt-original}) states that for a primitive matrix $W\in \Md$, we have \[p(W)\leq (d^2-2d+2).\]
The interesting part is that the above inequality only takes into account the dimension and does not depend on the matrix elements. Wielandt inequality has broad applications in graph theory and combinatorics, number theory and Markov chains (\cite{H-J},\cite{seneta}).

Sanz et al.(see \cite{Wielandt}) extended this concept of the classical Wielandt inequality to quantum channels (trace preserving and completely positive maps) and derived an upper bound on the number of iterations of a channel required to ensure that all the output density matrices must be of full rank. Their main result states that if $\Ep:\Md\rightarrow \Md$ is a primitive quantum channel with $n$ linearly independent Kraus operators, then the quantum Wielandt inequality or the quantum primitive index, denoted by $\omega(\Ep)$ satisfies the following inequality:
\begin{equation}\label{intro-ineq}
\omega(\Ep)\leq (d^2-n+1)d^2.
\end{equation}

Our work is motivated by one of the questions raised in the  Section VI in \cite{Wielandt} which asks for optimal bounds 
of primitivity index of positive maps as opposed to that of quantum channels. Since positive maps do not admit Kraus decompositions, any bound on primitivity index, must therefore be different from the bound given above in Equation \ref{intro-ineq}. Indeed, we show that for a primitive trace preserving positive map defined on $\Md$ which satisfies the Schwarz inequality, the primitivity index $\omega(\Ep)$ satisfies the following inequality:
\begin{equation}\label{intro-ineq-2}
\omega(\Ep)\leq 2(d-1)^2.
\end{equation}
Note that in this case, only dimension of the matrix algebra plays a role and not the linear map itself which is very similar to the spirit of the bound given in (classical) Wielndt's inequality. As positive maps play key role in detecting entanglement, a theory of primitivity bound for positive maps is required.

In a very recent work (see \cite{Mat-Shitov}) by Micha{\l}ek and Shitov, it has been shown that the bound given in \ref{intro-ineq} can be improved and they proposed a new bound of $O(d^2\log d)$ of the  Wielandt inequality. 
Our result remains independent of this new bound as their result is relevant assuming that the map admits a Kraus decomposition, that is, for completely positive maps.  
However, in light of this new result, our bound (in $O(d^2)$) for positive maps makes this research topic more interesting and when applied to completely positive maps it provides more evidence for the validity of a conjecture (conjecture 2) proposed in \cite{Garcia-etall}.                   
\section{Index of primitivity and quantum Wielandt bound}
We begin with some definitions and analyze closely the work of Sanz et al. The following definitions are the key concepts of the so-called non-commutative Perron-Fobinius theory. We refer to the articles \cite{evans-krohn}, \cite{wolf}, \cite{farenick}, \cite{Wielandt} for some preliminary background on this topic. 

Note that a positive linear map on a matrix algebra $\Md$ is one that sends every positive semidefinite elements of $\Md$ to positive semidefinite elements. Also recall that for two Hermitian matrices $a,b$, the relation $a\geq b$ means that $a-b$ is a positive semidefinite matrix.   
\begin{definition}{\label{def-irr}}
A positive linear map $\Ep:\Md\rightarrow \Md$ is called irreducible if $\Ep(p)\leq \lambda p$, for any projection $p\in \Md$ implies that $p=0$ or $p=1$.
\end{definition}
Given a linear map $\Phi:\Md\rightarrow\Md$, the spectrum of $\Phi$ which is denoted by $\rm{Spec(\Phi)}$, is defined as
\[\rm{Spec(\Phi)}=\{\lambda\in \mathbb{C}: (\lambda .id-\Phi) \ \text{is not invertible on} \ \Md\},\]
where $id$ denotes the identity operator on $\Md$. Recall that the spectral radius of $\Phi$ which is denoted as $r(\Phi)$, is defined as
\[r(\Phi)=\sup\{|\lambda|:\lambda\in \rm{Spec(\Phi)}\}.\]
It follows that if $\Phi$ is a unital positive map, then $r(\Phi)\leq 1$ and hence all eigenvalues lie in the unit disc of the complex plane (see Proposition 6.1 in \cite{wolf}). For a unital positive map, the set of eigenvalues located in the unit circle is called the peripheral spectrum of $\Phi$.

\begin{definition}
A unital positive linear map $\Ep:\Md\rightarrow\Md$ is called primitive if it is irreducible and moreover, $\rm{Spec}(\Ep)\cap\mathbb{T}=\{1\}$. This means that the only peripheral spectrum of $\Ep$ is the identity element 1. 
\end{definition}
\begin{definition}
A positive linear map $\Ep:\Md\rightarrow\Md$ is called strictly positive if it sends every \emph{positive semideifinite} element in $\Md$ to a \emph{positive definite} element.
\end{definition}
Note that whether a completely positive map (maps of the form $x\mapsto \sum_{j}a_jxa_j^*$) on $\Md$ is strictly positive or not is an NP-hard problem \cite{NP-hard}. Since a completely positive map is positive as well, checking whether a positive map is strictly positive is also an NP-hard problem.

We briefly mention here that a finite power of a primitive map must become a strictly positive map. It follows from the Perron-Frobenius theory of positive maps (see \cite{evans-krohn}) that for a unital irreducible map $\Ep$ on $\Md$, 1 is a non-degenerate eigenvalue of $\Ep$ and the corresponding eigenvector (call it $\rho_0$) is positive definite. Hence we have $\rho_0=1_d$ and the definition of primitivity above implies that  $\lim_{k\to \infty}\Ep^k(\rho)=\tr(\rho)1_d:=P_{\infty}(\rho)$. Here $1_d$ is the identity matrix in $\Md$ and $P_\infty$ is an idempotent positive map. Now following Proposition 6.7 in \cite{wolf}, for any $k\in \mathbb{N}$ and $\rho\in \Md$ if $\Ep^k(\rho)$ has a vector $\psi$ such that $\Ep^k(\rho)\psi=0$, then we have
\[1=|\langle \Ep^k(\rho-1_d)\psi,\psi\rangle|\leq ||\Ep^k(\rho)-1_d||=||(\Ep^k-P_\infty)(\rho-1_d)||\leq \mu^k c||\rho-1_d||, \]
where $0<\mu<1$ and $0<c$ is a constant dependent of $\Ep$ but independent of $k$. Taking $k$ sufficiently large, the above inequality gives a contradiction for any $\rho\in \Md$. Hence there must exists a $k\in \mathbb{N}$ such that $\Ep^k(\rho)$ is positive definite for any positive semidefinite $\rho\in \Md$.

\begin{definition} {\rm[{see \cite{Wielandt}}]}
For a primitive positive map $\Ep:\Md\rightarrow\Md$, the {\rm\textbf{index of primitivity} (denoted by $\omega(\Ep)$)} is the least natural number $k$, such that $\Ep^k(a)$ is positive definite for every positive semidefinite $a\in \Md$.   
\end{definition}
The index of primitivity as defined above is a generalization of classical primitivity index of a non-negative primitive matrix.
\begin{theorem}{\rm[{see \cite{Wielandt}}]}
Let $\Ep:\Md\rightarrow\Md$ be a primitive quantum channel with $n$ linearly independent Kraus operators, that is, there are linearly independent elements $\{a_i:1\leq i\leq n\}$ such that 
\[\Ep(x)=\displaystyle\sum_{i=1}^n a_i x a_i^*.\]
Then \[\omega(\Ep)\leq (d^2-n+1)d^2.\]
\end{theorem}

The above bound is called the quantum version of the \textbf{Wielandt bound} which is a generalization of the classical Wielandt number of primitive matrices. 
The bound $(d^2-n+1)d^2$ was obtained by looking at the following quantity:
\[i(\Ep)=\min\{k\in \mathbb{N}: \rm{Span}\{a_{i_1}\cdots a_{i_k}\}=\Md\},\]
which can be viewed as the minimum number $k$ for which the 
the Choi matrix of $\Ep^k$ has full rank, that is, $\rm{rank}(C_{\Ep^k})=d^2$. It was proved in \cite{Wielandt} that $\omega(\Ep)\leq i(\Ep)$ and then it was proved that $i(\Ep)\leq (d^2-n+1)d^2 $. In this context, it is worth pointing out again the work of Micha{\l}ek and Shitov \cite{Mat-Shitov} who used algebraic techniques to improve this bound from $O(d^4)$ to $O(d^2\log d)$. This new bound of course applies to primitive quantum channels and more generally to primitive completely positive maps.

For finding a bound of the primitive index for positive maps we can not follow the above procedure because our maps do not admit any Kraus decomposition and hence a very different approach must be taken to achieve this goal. We discuss this approach in the following section.
 

\section{A new Wielandt bound for positive maps}
We begin with generalizing the concept of irreducibility of a linear map. Note that two projections $p,q\in \Md$ are said to be (Murray-von Neumann) equivalent (written as $p\sim q$) if there is an operator $v\in \Md$ such that $vv^*=p$ and $v^*v=q$. It follows that $p\sim q$ if and only if $\tr(p)=\tr(q)$.
\begin{definition} A positive linear map $\Ep:\Md\rightarrow\Md$ is defined to be 
\textbf{fully irreducible} if $\Ep(p)\leq \lambda q$, for two projections $p,q$ with $p\sim q$ and $\lambda>0$ implies $p,q\in \{0,1\}$.
\end{definition}
Note that the above definition arose in \cite{idel} and these maps were called fully indecomposable however we are avoiding this terminology because there is a concept of indecomposibility in the theory of positive maps. The above definition clearly generalizes the irreducibility (see Definition \ref{def-irr}) so we will call such maps \emph{fully irreducible}. It is evident that fully irreducibility is stronger than irreducibility as in the later case, trivially one can put $q=p$ and vacuously $p\sim p$. We begin with a lemma.
\begin{lemma}For a positive element $a\in \Md$, if for a projection $p$ we have $pap=0$, then $(1-p)ap=0=pa(1-p)$.
\end{lemma}
\begin{proof}
If $\xi\in \rm{Range(p)}$ and $\eta\in \rm{Range(p)}^\perp$, then $p\xi=\xi$ and $p\eta=0$. By positivity of $a$, we have for every $\lambda\in \mathbb{C}$
,
\[\langle a(\lambda \xi+\eta),\lambda \xi +\eta
\rangle\geq 0.\]
Now using $pap=0$, we get from the above inequality $\langle a\eta, \eta\rangle + 2\rm{Re}(\lambda\langle a \xi,\eta\rangle)\geq 0$. This implies $\langle a\xi,\eta\rangle =0$. This yields
\[\langle (1-p)ap (\xi+\eta), \xi+\eta\rangle=\langle ap(\xi+\eta), \xi+\eta \rangle =0.\]
Similarly $pa(1-p)=0$.
\end{proof}

\begin{proposition}\label{prop-full-irr}
For a unital positive map $\Ep$, if $\Ep(p)\leq \lambda q$,  for some $\lambda>0$ and projections $p,q$, then we have $\Ep(p)\leq q$. 
\end{proposition}
\begin{proof}
We first note that $\Ep(p)\leq \lambda q$, implies that 
$a\Ep(p)a^*\leq a(\lambda q)a^*$, for any $a\in \Md$. Using $a=1-q$, we get
\[
0\leq (1-q)\Ep(p)(1-q)\leq (1-q)(\lambda q)(1-q)=0.\]
Hence we obtain
\begin{equation}\label{eq-1}
(1-q)\Ep(p)(1-q)=0
\end{equation}
Now by the previous lemma we have 
\[q\Ep(p)(1-q)=0=(1-q)\Ep(p)q.\]
These equations result in \[q\Ep(p)q=q\Ep(p)=\Ep(p)q.\]
Now expanding the Equation \ref{eq-1} we get
\[0=(1-q)\Ep(p)(1-q)=\Ep(p)-\Ep(p)q-q\Ep(p)+q\Ep(p)q=\Ep(p)-q\Ep(p)q.\]
Now using the unitality of $\Ep$ we know that $\|\Ep\|=1$ and we obtain  
\[\Ep(p)=q\Ep(p)q\leq q \|\Ep(p)\|1q\leq q.\] 
\end{proof}
\begin{remark}
Intuitively, the above proposition provides a relaxation in finding a projection $p$ in the set $F=\{a\in \mathcal{M}_{d+}: \Ep(a)\leq \lambda q\}$, where $q$ is a given projection and $\mathcal{M}_{d+}$ is the set of all positive semidefinite elements of $\Md$. Indeed, if we let $F_1=F\bigcap\{x\in \Md: ||x||\leq 1\}$, then the above Proposition says that if $p\in F$, then $p\in F_1$. Notice that the set $F$ is a face of the convex set $\mathcal{M}_{d+}$, that is, if $a\geq 0$ and $b\in F$, then $a\leq b$ implies that $a\in F$. This relaxation holds because a unital positive map maps the unit ball 
$\mathbf{B}=\{x\in \Md: ||x||\leq 1\}$ to itself.
\end{remark}
\begin{proposition}\label{prop-proj-eqi}
For a unital trace preserving map $\Ep:\Md\rightarrow \Md$, $\Ep$ is fully irreduicble if and only if there are  no nontrivial projections $p,q$ with $p\sim q$ such that $\Ep(p)=q$.
\end{proposition}
\begin{proof}
The `if' part is obvious as $\Ep(p)=q$ clearly violates the definition of fully irreducibility. Conversely, suppose $\Ep$ is not fully irreducible. Then there are projections $p,q$ and $p\sim q$ such that $\Ep(p)\leq \lambda q$ which by proposition \ref{prop-full-irr} we have $\Ep(p)\leq q$. Now using the trace preservation of $\Ep$ and the faithfulness of trace, we get $\Ep(p)=q$.
\end{proof}

We note down an observation here that every unital and trace preserving positive map is rank increasing.
\begin{proposition}\label{rank-increase}
Let $\Ep:\Md\rightarrow \Md$ a unital and trace preserving positive map. Then if $a$ is a positive element in $\Md$, then it follows that
\[\rm{Rank}(\Ep(a))\geq \rm{Rank}(a).\] 
\end{proposition}
\begin{proof}
It is a consequence of Uhlmann's theorem (see Theorem 4.33 in \cite{Watrous-book}) that $b=\Phi(a)$, for a unital and trace preserving positive map $\Phi$ if and only if 
\[\lambda(b)\prec \lambda(a),\]
that is the vector of eigenvealues of $a$ majorizes the vector of eigenvalues of $\Phi(a)$.

Now it is enough to prove the proposition for projections.
If $p$ is a projection, then by Uhlmann's theorem 
$\lambda(\Ep(p))\prec \lambda(p)$. Then it follows that
\[\rm{Rank}(\Ep(p))\geq \rm{Rank}(p).\]      
\end{proof}
A unital positive map $\Phi:\Md\rightarrow\Md$ satisfies the following inequality (see \cite{stormer}, Theorem 1.3.1)
\[\Phi(aa^*)\geq \Phi(a)\Phi(a^*),\]
for all elements $a$ satisfying $aa^*=a^*a$. Using the Schwarz inequality for positive maps on hermitian elements, we can derive a stronger result for fully irreducible maps. 
The following result first appeared as Proposition 1.25 in \cite{idel}. Here we give a different proof. By singular positive elements we mean positive elements which have at least one zero eigenvalue.
\begin{theorem}\label{thm-strict-kernel}
Let $\Ep:\Md\rightarrow\Md$ be a unital and trace preserving positive map. Then $\Ep$ is fully irreducible if and only if for all singular positive elements $a\in \Md$, we have \[\rm{Rank} \ (\Ep(a))> \rm{Rank}(a),\] that is, $\Ep$ is fully irreducible if and only if it is strictly rank increasing or equivalently it is strictly kernel reducing.
\end{theorem}
\begin{proof}
A unital positive linear map satisfies the Schwarz inequality on normal elements. So using this inequality on a positive element $a^{1/2}$ one gets 

\[\Ep(a)=\Ep(a^{1/2}a^{1/2})\geq \Ep(a^{1/2})\Ep(a^{1/2}).\]
As $x\mapsto \sqrt{x}$ is an operator monotone function we get $\Ep(a)^{1/2}\geq \Ep(a^{1/2})$. We can continue this process to get 
\[\Ep(a)^{1/2^n}\geq \Ep(a^{1/2^n}), \ \forall n.\]
If $q$ is the projection onto the $\rm{Range}(\Ep(a))$ and $p$ is the projection onto $\rm{Range}(a)$ then note that as $a$ is singular, $p$ is non-trivial. Now by the spectral theorem for positive elements, taking limit as $n\to\infty$ in the above equation we obtain
\[q\geq \Ep(p).\]
Now if $a$ and $\Ep(a)$ have the same rank, then $p\sim q$ and this violates the fully irreducibilty property. 

Conversely, suppose there are non-trivial projections $p,q$ with $p\sim q$ such that $\Ep(p)\leq \lambda q$, for $\lambda>0$. Using Proposition \ref{prop-full-irr} we have $\Ep(p)\leq q$. 
Then $q-\Ep(p)\geq 0$. Using the trace preservation property of $\Ep$ and faithfulness of trace we get $q=\Ep(p)$. Since \[\rm{rank} (p)=Tr(p)=Tr(\Ep(p))=Tr(q)=\rm{rank}(q),\]
it violates the (strictly) rank increasing property.
\end{proof}
Now we introduce one more concept related to a linear map acting on $\Md$.
\begin{definition}
The multiplicative domain $\mathcal{M}_\Phi$ of a linear map $\Phi:\Md\rightarrow\Md$ is the following set:
\[\mathcal{M}_\Phi=\{a\in \Md: \Phi(ab)=\Phi(a)\Phi(b), \Phi(ba)=\Phi(b)\Phi(a) \ \forall b\in \Md \}.\]
\end{definition}
\begin{definition}
We say a positive linear map $\Phi:\Md\rightarrow\Md$ is a Schwarz map if it satisfies the Schwarz inequality $\Phi(aa^*)\geq \Phi(a)\Phi(a^*)$, for every element $a\in \Md$.
\end{definition}
It is a consequence of the Stinespring dilation theorem for completely positive maps that every unital completely positive map acting on a C$^*$-algebra satisfies the Schwarz inequality. However, following \cite{choi1980b}, the map $\Phi:\mathcal{M}_2\rightarrow \mathcal{M}_2$  
defined by
\[\Phi\Big(\begin{bmatrix} 
x_{11} & x_{12}\\
x_{21} & x_{22}
\end{bmatrix}\Big)=\frac{1}{2}\begin{bmatrix}
x_{11}+\frac{x_{11}+x_{22}}{2} & x_{21}\\
x_{12} & x_{22}+\frac{x_{11}+x_{22}}{2}
\end{bmatrix},\]
is a Schwarz map but fails to be 2-positive.
For a Schwarz map $\Phi$, the set $\mathcal{M}_\Phi$ is a C$^*$-subalgebra of $\Md$ (see Corollary 2.1.6 in \cite{stormer}). Following \cite{miza}, given a Schwarz map $\Phi$ on $\Md$, one obtains a decreasing chain of C$^*$-subalgebras 
\[\mathcal{M}_{\Phi}\supseteq\mathcal{M}_{\Phi^2}\supseteq\cdots
\supseteq\mathcal{M}_{\Phi^n}\supseteq \cdots.\]
For finite dimensionality, the above chain stabilizes to the subalgebra
\[\mathcal{M}_{\Phi^{\infty}}=\bigcap_{n\geq 1}\mathcal{M}_{\Phi^n}.\]
The minimum number $n$ required for the channel to reach to this subalgebra $\mathcal{M}_{\Phi^{\infty}}$ is called the multiplicative index and denoted by $\kappa(\Phi)$.
\begin{remark}
It should be noted here that the stabilized multiplicative domain $(\mathcal{M}_{\Phi^{\infty}})$ and the multiplicative index $(\kappa(\Phi))$ of a linear map $\Phi$ can be defined as long as the map is a Schwarz map. Indeed the map need not be a channel as these concepts originated (\cite{miza}) exploiting only the Schwarz inequality and trace preservation property of $\Phi$.
\end{remark}
\begin{proposition}\label{fully-irrred-trivial-mult.}
 A trace preserving Schwarz map $\Ep:\Md\rightarrow\Md$ is fully irreducible if and only if it has trivial multiplicative domain, that is, $\Me=\mathbb{C}1$. 
\end{proposition}
\begin{proof}
First, we note that a trace preserving Schwarz map is unital. Indeed, $\Ep(xx^*)\geq \Ep(x)\Ep(x^*)$ implies \[\|\Ep(x)\|^2=||\Ep(x)\Ep(x^*)||\leq ||\Ep(xx^*)||\leq \|\Ep\|\|x\|^2.\] Using the Russo-dye  theorem (Corollary 2.9 in \cite{paulsen}) we get from the above inequality with $x=1$,
\[\|\Ep\|^2\leq \|\Ep\| \ \Rightarrow \|\Ep\|\leq 1.\]
As $\Ep$ is a contraction, we get $\Ep(1)\leq 1$ and hence by trace preservation we get $\Ep(1)=1$.

If $\Me$ is not trivial, then there exists a projection $p\in \Me$. Now by definition of multiplicative domain, $\Ep(p)$ is again a projection, call it $q$. Using the trace preserving property we get $p\sim q$. This is a contradiction following Proposition \ref{prop-proj-eqi}. The converse follows exactly in the similar way.  
\end{proof}
It should be notes here that being fully irreducible or equivalently having $\Me=\mathbb{C}1$, does not force the map to be strictly positive. The following example verifies this fact.
\begin{example}
Consider the map $\Ep:\mathcal{M}_3\rightarrow\mathcal{M}_3$ defined by
\[\Ep(x)=\frac{1}{2}(\tr(x)1-x^t).\]
It is easy to verify that $\Ep$ is unital and trace preserving Schwarz map. It follows that any rank one projection is mapped to a rank 2 element, so rank one projections can not be in the multiplicative domain. Since this domain is a unital C$^*$-subalgebra, it follows that there is no rank 2 projection in the multiplicative domain as well. Hence $\Me=\mathbb{C}1$.  
 Now it is easily seen that
the image of the rank one matrix unit $E_{11}$, $\Ep(E_{11})=\frac{1}{2}(E_{22}+E_{33})$ which is of rank $2(\neq 3)$.
\end{example} 

 We are ready to state and prove the main theorem of this article.
\begin{theorem}\label{main thrm}
Let $\Ep:\Md\rightarrow\Md$ be a trace preserving primitive Schwarz map with the multiplicative index $\kappa$. Then 
$\Ep^{\kappa(d-1)}$ sends every positive semi definite matrix to a positive definite matrix. That is,
\[\omega(\Ep)\leq \kappa(\Ep)(d-1).\]
\end{theorem}
\begin{proof}
As $\Ep$ is a trace preserving Schwarz map, it is unital. 
Following the Corollary 3.5 in \cite{miza}, $\Ep$ is primitive implies that $\Mi=\mathbb{C}1$.
First, note the following chain has length $\kappa$:
\[\Me\supseteq\mathcal{M}_{\Ep^2}\supseteq\cdots
\supseteq\Mi=\mathbb{C}1.\]
First of all observe that $\kappa(\Ep)\leq \omega(\Ep)$. 
This is because if we take a projection $p\in \Me$ such that $\Ep(p)\in \Me$, then $\Ep(p)$ is again a projection. By the definition of $\mathcal{M}_{\Ep^2}$ (see \cite{miza}), $p\in \mathcal{M}_{\Ep^2}$. If $\mathcal{M}_{\Ep^2}$ is still not $\mathbb{C}1$, we get $\Ep(p)$ is not positive definite. Repeating the argument for $\Ep^2, \Ep^3, \cdots\Ep^{(\kappa-1)}$ we see that if $p\in \mathcal{M}_{\Ep^{\kappa -1}}$, then $\Ep^{\kappa}(p)\in \mathbb{C}1$ which then makes it invertible and hence $\Ep^{(\kappa)}$ maps every projection in $\Me$ to an invertible operator. Thus $\kappa(\Ep)\leq \omega(\Ep)$. 

Now we will show $\omega(\Ep)\leq \kappa(\Ep)(d-1)$. 
Strict positivity of any map $\Phi$ will be guaranteed if $\Phi(p)$ is invertible for any rank one projection $p$.
Indeed, Given any projection $q$, there exists a rank one projection $p$ such that $p\leq q$. Hence for any positive linear map $\Phi$, $\Phi(p)\leq\Phi(q)$. So if $\Phi(p)$ is invertible, then so is $\Phi(q)$. Now if the spectral decomposition of  positive element $a\in \mathcal{M}_{d}$ be given by
\[a=\displaystyle \sum_{j=1}^{k}\lambda_{j}p_{j},\]
where $p_{j}$'s are spectral projections onto the the eigenspace corresponding to the eigenvalue $\lambda_{j}$, then $\Phi(a)=\displaystyle\sum_{j=1}^{k}\lambda_{j}
\Phi(p_{j})$. Now if for every $\xi\in \mathbb{C}^{d}$, $\langle\Phi(p_{j})\xi,\xi\rangle>0$, then 
$\langle\Phi(a)\xi,\xi\rangle>0$.
Hence it is enough to show that for any rank one projection $p$, $\Phi(p)$ is invertible.

Take a rank one projection $p\in \Md$. Since $\Ep^{\kappa(\Ep)}$ has multiplicative domain $\Mi=\mathbb{C}1$, considering the unital map $\Phi \ (=\mathcal{E}^{\kappa(\Ep)})$ we see that it has trivial multiplicative domain. Now by Proposition \ref{fully-irrred-trivial-mult.} and Theorem \ref{thm-strict-kernel}, it is strictly kernel reducing, that is,
\[\rm{dim \ Ker}(\Phi(a))<\rm{dim \ Ker}(a), \ \forall \ a\in \mathcal{M}_{d+}.\]
Here $\mathcal{M}_{d+}$ denotes the set of all positive semidefinite elements of $\Md$.
Now taking a rank one projection $p$, we evaluate 
\[\rm{dim \ Ker}\Phi^{(d-1)}(p)<\rm{dim \ Ker}\Phi^{(d-2)}(p)<\cdots<\rm{dim \ Ker}\Phi(p)<\rm{dim \ Ker}(p) .\]
Since $p$ has rank 1, the kernel has dimension $d-1$ and since the dimension is a non negative integer function, the above inequality yields \[\rm{dim \ Ker}\Phi^{(d-1)}(p)=0.\] 
 Hence $\Phi^{(d-1)}(p)$ is invertible. Since this holds for every rank one projection, $\Phi^{d-1}$ is strictly positive and hence $\Ep^{\kappa(\Ep)(d-1)}$ is strictly positive.  
\end{proof}
Now it is important to find a suitable bound for $\kappa(\Ep)$ for a trace preserving Schwarz map. In \cite{Sam-Miza}, Theorem 3.6 such a bound was put forward. The key point is that the bound was obtained by utilizing the C$^*$-algebra structure of the subalgebras $\Me, \mathcal{M}_{\Ep^2}$ etc. As for a Schwarz map $\Ep$, these subalgebras are all C$^*$-algebras, we can use this bound in our context.
\begin{corollary}\label{cor-main-bound}
For a trace preserving primitive Schwarz map $\Ep$ acting on $\Md$, we have 
\[\omega(\Ep)\leq 2(d-1)^2.\]
\end{corollary}
\begin{proof}
 As was shown in \cite{Sam-Miza}, $\kappa(\Ep)$ must be less than the maximum length of the chain of subalgebras 
 \[\mathcal{M}_{\Ep}\supseteq\mathcal{M}_{\Ep^2}\supseteq\cdots
\supseteq\mathcal{M}_{\Ep^n}\supseteq \Mi.\]

 It follows that $\kappa(\Ep)\leq 2(d-1)$. So it shows that the Weilandt number $\omega(\Ep)\leq 2(d-1).(d-1)$.
 \end{proof}
 \begin{remark}
 It should be noted here that similar to the bound given in 
 \cite{Wielandt}, we don't know whether the inequality given in Corollary \ref{cor-main-bound} is sharp or not. Even for a quantum channel, the optimal value for $\kappa$ 
is still unknown and hence deciding whether a positive map attains the exact Wielandt bound is an avenue for future research. It was also brought to our attention through private communication by Micha{\l} Bia{\l}o\' nczyk that our proof in Theorem \ref{main thrm} works for arbitrary positive linear maps without the Schwarz inequality assumption. Although the notion of multiplicative domain is 
same for positive maps as Schwarz maps, this domains are no longer C$^*$-algebras. Indeed, they are Jordan algebras (see \cite{stormer2007}). To get a bound for the multiplicative index in this case, one needs to get a maximum bound for a chain of Jordan algebras as in the proof of the Corollary \ref{cor-main-bound}. This again a possibility for future research.  
 \end{remark}
We utilize these findings to quantum channels to get better bounds of some classes of channels. We recall that a linear map $\Phi$ on $\Md$ is called PPT(positive partial transpose) if $\Phi$ and $\Phi\circ t$ are completely positive, where `t' denotes the transpose map $x\mapsto x^t$. Also, a linear map $\Phi$ on $\Md$ is entanglement breaking if $(id_n\otimes \Phi)(\Gamma)$ is always a separable state for any bipartite state $\Gamma\in \mathcal{M}_n\otimes \Md$. For more properties of these maps see \cite{stormer}, \cite{entng-brkng}.
\begin{corollary}
For unital PPT and entanglement breaking channels $\Ep$ acting on $\Md$, we have
 \[\omega(\Ep)\leq d(d-1).\] 
\end{corollary}
\begin{proof}
We know that the unital PPT channels and the entanglement breaking channels have abelian multiplicative domain (see \cite{R-J-P}).
It is not hard to see that the multiplicative index of these channels can be maximum $d$ (see Proposition 3.2 in \cite{Sam-Miza}). So \[\omega(\Ep)=\kappa(\Ep)(d-1)\leq d(d-1).\]
\end{proof}
\begin{remark}
Since $d(d-1)< d^2$, we have a better bound of Weilandt inequality than that given in \cite{Wielandt} for PPT and entanglement breaking channels.
\end{remark} 
\begin{proposition}
Let $\Ep:\Md\rightarrow\Md$ be a trace preserving primitive map such that $\Ep$ and the adjoint map $\Ep^*$ satisfies the Schwarz inequality. Then $\omega(\Ep^*)\leq 2(d-1)^2$.
\end{proposition}
\begin{proof}
Since trace preserving Schwarz maps are unital and the adjoint map of a unital map is trace preserving, we will have $\Ep^*$ is unital and trace preserving Schwarz map.
From \cite{Sam-Miza}, $\kappa(\Ep)=\kappa(\Ep^*)$. Also if $\Ep$ is primitive, then so is $\Ep^*$. Hence the assertion follows from the Theorem \ref{main thrm}.
\end{proof}
\section{Wielandt bound for tensor product channels }
   
Since the Wielandt bound as given in Theorem \ref{main thrm} involves the multiplicative index, we can get a handle of Wilandt inequality for tensor products of channels. This is possible because multiplicative domain 
(and hence multiplicative index) of tensor products of unital channels behave nicely. We state this result below:

\begin{theorem}[{\rm{See}\cite{Sam-Miza}}]\label{sam-miza}
If $\Phi,\Psi$ are two unital channels on $\Md$, then the multiplicative domain of $\Phi\otimes \Psi$ splits, that is, 
\[\mathcal{M}_{\Phi\otimes\Psi}=\mathcal{M}_{\Phi}\otimes
\mathcal{M}_{\Psi}.\]
Moreover, \[\kappa(\Phi\otimes\Psi)=\rm{max}\{\kappa(\Phi),\kappa(\Psi)\}=\rm{max}\{2(d_1-1), 2(d_2-1)\}.\]
\end{theorem}
Using the  Theorem \ref{main thrm}, we immediately get
\begin{proposition}
For unital primitive channels \[\Phi:\mathcal{M}_{d_1}\rightarrow\mathcal{M}_{d_1} \ \text{and} \ \Psi:\mathcal{M}_{d_2}\rightarrow\mathcal{M}_{d_2}\] we have the Weilandt bound 
\[\omega(\Phi\otimes\Psi)=\rm{max}\{\omega(\Phi),\omega(\Psi)\}\leq\rm{max}\{2(d_1-1)^2,2(d_2-1)^2\}.\]
\end{proposition}
\begin{proof}

First of all note that the tensor product of two primitive maps is primitive. Indeed, this fact was proved in \cite{Sam-Miza}, Theorem 2.10 utilizing the splitting property:
\[\mathcal{M}_{(\Phi\otimes\Psi)^\infty}=\mathcal{M}_
{\Phi^\infty}\otimes\mathcal{M}_{\Psi^\infty}.\]
Hence the result follows immediately from Theorem \ref{main thrm}.
\end{proof}

\section{Dichotomy result for the zero-error quantum capacity}
Similar to one given in \cite{Wielandt}, we can establish  a dichotomy result for unital channels with respect to the quantum capacity using the Wielandt bound.
\begin{definition}
The \textbf{one shot zero-error classical capacity} $(C_0(\Phi))$ of a channel $\Phi$ is defined to be  $\displaystyle\sup_{S\in\mathfrak{S}}\log|S|$, where $\mathfrak{S}$ is the  set of all families of density matrices $\{\rho_i\}$ such that $\tr(\Phi(\rho_i)\Phi(\rho_j))=0$ for $i\neq j$.
 \end{definition}
 Consider the following dichotomy theorem for one shot zero-error classical capacity of channels:
\begin{theorem}[\rm{Sanz-Garc\'ia-Wolf-Ciraq}, \cite{Wielandt}]\label{dichotomoy-wolf}
If $\Ep$ is a quantum channel with a full rank fixed point, then either $C_0(\Ep^n)\geq 1$ for all $n$ or $C_0(\Ep^{\omega(\Ep)})=0$. Here $\omega(\Ep)$ is the Wielandt bound for $\Ep$.
\end{theorem} 
 
 Now we will prove a dichotomy theorem for another capacity of quantum channel.
 \begin{definition}{(see \cite{Shir-Tat})}
 Let $\Phi:\mathcal{B}(\mathcal{H})\rightarrow \Bh$ be a channel where $\rm{dim(\mathcal{H})}<\infty$. Then the \textbf{one shot zero-error quantum capacity} $(Q_0(\Phi))$ of a channel $\Phi$ is defined to be 
 $\displaystyle\sup_{\mathcal{K}\in \mathfrak{C}}\log \rm{dim(\mathcal{K})}$, where $\mathfrak{C}$ is the collection of subspaces $\mathcal{H}_0$ of $\mathcal{H}$ such that there exists a channel $\Psi$, satisfying $\Psi(\Phi(\rho))=\rho$, for all $\rho$ supported on $\mathcal{H}_0$.
 \end{definition}
 Now we write down the dichotomy theorem for this capacity of channel:
 \begin{theorem}\label{dichotomy-new}
 Let $\Ep$ be a unital channel on $\Md$. Then $Q_0(\Ep^n)>0$ for all $n$ or $Q_0(\Ep^{\omega(\Ep)})=0$.  
 \end{theorem}
 \begin{proof}
 Let the Kraus representation of $\Ep$ be given by $\Ep(x)=\displaystyle \sum_{j=1}^m a_j x a_j^*$.
 We encounter two mutually exclusive situations: either $\Ep$ is primitive or not.
 
 \textbf{Case 1:}
 
 Suppose $\Ep$ is primitive. So $\Ep^{\omega(\Ep)}$ sends every density operators to density operators  with full rank.  Now if $Q_0(\Ep^{\omega(\Ep)})>0$, there is a proper-subspace $\mathcal{H}_0$ and channel $\Psi$ such that $\Psi(\Ep^{\omega(\Ep)}(\rho))=\rho$ for $\rho$ supported on $\mathcal{H}_0$. Observe that the Kraus representation of $\Ep^{\omega(\Ep)}$ is given by \[\Ep^{\omega(\Ep)}(x)=\sum_{i_1,\cdots, i_{\omega(\Ep)}=1}^m
 a_{i_1}\cdots a_{i_{\omega(\Ep)}}xa_{i_{\omega(\Ep)}}^*\cdots a_{i_1}^*.\]
Now note that by the Knill-Laflamme (\cite{knill-lafl}) condition of reversibility of the channel $\Ep^{\omega(\Ep)}$ is equivalent to the condition that, $\forall \xi,\eta\in \mathcal{H}_0$ with $\langle \xi,\eta\rangle=0$ and all $j_i,\cdots,j_{\omega(\Ep)},i_1,\cdots, i_{\omega(\Ep)}\in\{1,2,\cdots m\}$, implies that 
\begin{equation}\label{knill-laflamme}
\langle \xi, (a_{j_1}^*\cdots a_{j_{\omega(\Ep)}}^* a_{i_1}\cdots a_{i_{\omega(\Ep)}})\eta\rangle=0.
\end{equation}
Now if $\Ep^{\omega(\Ep)}$ is strictly positive, then so is $\Ep^{*\omega(\Ep)}\circ\Ep^{\omega(\Ep)}$. Indeed, for any unit vector $\psi( q=\psi\psi^*)$ and rank one projection $p=\phi\phi^*$, we have 
\[\langle\Ep^{*\omega(\Ep)}\circ\Ep^{\omega(\Ep)}(p)\psi,\psi\rangle=\tr(\Ep^{*\omega(\Ep)}\circ\Ep^{\omega(\Ep)}(p)q)=\tr(\Ep^{\omega(\Ep)}(p)\Ep^{\omega(\Ep)}(q)) >0.\]
Hence \[\text{Span}\{a_{j_1}^*\cdots a_{j_{\omega(\Ep)}}^* a_{i_1}\cdots a_{i_{\omega(\Ep)}}\}=\Md.\]
Clearly this violates the Equation \ref{knill-laflamme}.

\textbf{Case 2:}

Suppose $\Ep$ is not primitive. So $\Mi$ is non-trivial. 
Following Theorem 2.5 in \cite{miza} we get $\Md=\Mi\oplus \Mi^{\perp}$ and $\Ep$ is an automorphism on $\Mi$ with inverse being $\Ep^*$. It follows that $\Ep^n$ also is an automorphism 
on $\Mi$ of every $n$ with the inverse $\Ep^{*n}$. Now as $\Mi$ is non-trivial algebra, there exists a projection $p\in \Mi$ whose support is $\mathcal{H}_0$ say. Then for every $n\in \mathbb{N}$, we have $\Ep^{*n}\circ\Ep^n(p)=p$. So for every $n$, there is a recovery channel $\Ep^{*n}$ for $\Ep^{n}$ and hence $Q_0(\Ep^n)>0$.
 \end{proof}
 \begin{remark}
 The dichotomy result for classical capacity given in the Theorem \ref{dichotomoy-wolf} works for channel with full rank fixed point(example-unital channels). Although Theorem \ref{dichotomy-new} deals with unital channels, it can be shown that any fully irreducible channel, when properly scaled, can be made into a unital channel (see \cite{gurvits}). So in relevant contexts one does not loose much by choosing to work with unital channels.
 \end{remark}

 \section{Wielandt bound and strictly contractive channels}
Strictly contractive channels were first introduced in \cite{Raginsky} where (strict)contractivity of channels with respect to the metric induced by the trace norm $(\|\cdot\|_{1})$ was considered. Later, in \cite{Doug-Miza} these contractions were studied with respect to the Bures metric
from a more operator algebraic viewpoint. For convergence analysis and entropy production of bi-stochastic channels, these strictly contractive maps are key objects to look at (\cite{ergodic},\cite{M-S-D-W}). 
\begin{definition}
A channel $\Phi$ is said to be strictly contractive if 
\[\|\Phi(\rho)-\Phi(\sigma)\|_{1}\leq c(\Phi)\|\rho-\sigma\|_{1}\]
for all density matrices $\rho,\sigma$ in $\Md$, with $\rho\neq \sigma$ and $0\leq c(\Phi)<1$.
\end{definition}
The constant $c(\Phi)$ is called the contractive modulus of $\Phi$. The $\|\cdot\|_{1}$-norm is defined by \[\|x\|_{1}=\tr[(xx^*)^\frac{1}{2}],\] for all $x\in \Md$. It can be proved that the strictly contractive channels are primitive. The converse follows once the primitive map becomes strictly positive. We let $\Omega$ to be the map defined on $\Md$ as follows 
\[\Omega(x)=Tr(x)\frac{1}{d}.\]
This map is a unital, trace preserving and (completely)positive linear map. We will use the variational definition (see Lemma 3.64 in \cite{math-language-qit}) of $||\cdot||_1$ given as 
\[||x||_1=\sup_{u\in \Md, uu^*=1}|Tr(xu)|, \forall x\in \Md.\]

We need the following result to prove our main theorem of the section. Here by ``interior" point, we just mean the relative interior of a convex cone. The essential feature of the next result is to assert that positive linear maps that are strictly positive, lie in the interior of the cone of positive maps. 

We recall (See section 2 in \cite{KYE-96}), that a point $x$ inside a convex set $C$ is said to be an \textit{interior point} of $C$ if for each $y\in C$, there exists a $t>1$ such that $(1-t)y+tx\in C$.  
\begin{proposition}[\rm{See Prop. 2 in \cite{farenick} and Prop 4.1 in \cite{KYE-96}}]\label{Prop-int}
Let $\Phi$ be a positive linear map on $\Md$. Then the following statements are true:
\begin{enumerate}
\item $\Phi$ sends every positive semidefinite matrices in $\Md$ to positive definite ones.
\item $\Phi$ sends every rank one projections to invertible positive matrices.
\item $\Phi$ is in the interior of the cone of all positive linear maps on $\Md$.
\end{enumerate}
\end{proposition}
\begin{theorem}\label{thm-strictly contrctv}
Given a primitive channel $\Ep:\Md\rightarrow\Md$, $\Ep^{\omega(\Ep)}$ is a strictly contractive map. 
\end{theorem}
\begin{proof}
Note that since $\Ep^{\omega(\Ep)}$ sends every positive semidefinite matrices to positive definite matrices, by Proposition \ref{Prop-int},   $\Ep^{\omega(\Ep)}$ lies in the interior of the convex set of positive maps. This means that there exists a $t>1$ such that $(1-t)\Omega + t \Ep^{\omega(\Ep)}$ is a positive map, call it $\Psi$. Here $\Omega$ is the completely depolairizing channel introduced in the beginning of this section. 
 Hence letting $\delta=t-1>0$, we obtain 
\[\Psi=(1+\delta)\Ep^{\omega(\Ep)}-\delta \Omega.\] This in turn means 
\[\Ep^{\omega(\Ep)}=\frac{1}{1+\delta}\Psi+\frac{\delta}{1+\delta}\Omega.\]
Note also that $\Psi$ is a trace preserving positive map. 
Now we show that a trace preserving positive map $\Psi$ decreases the $||\cdot||_1$ norm. Indeed, for any $x\in \Md$
\begin{align*} 
||\Psi(x)||_1 &=\sup _{u\in \Md, uu^*=1}|Tr(\Psi(x)u)|\\
&=\sup _{u\in \Md, uu^*=1}|Tr(x\Psi^*(u))|\\
&\leq \sup _{u\in \Md, uu^*=1}||\Psi^*(u)||.||x||_1 \\
&\leq 1.||x||_1
\end{align*}
Here we used the fact that the adjoint $\Psi^*$ of the trace preserving positive map $\Psi$ is a unital positive map and hence $||\Psi^*(u)||\leq 1$, for any unitary $u\in \Md$. Hence for any two density matrices $\rho,\sigma$, we get $||\Psi(\rho-\sigma)||_1\leq ||\rho-\sigma||_1$. Now we calculate for any densities $\rho\neq \sigma$, 
\begin{align*}
\|\Ep^{\omega(\Ep)}(\rho-\sigma)\|_{1}&=\|(\frac{1}{1+\delta}\Psi+\frac{\delta}{1+\delta}\Omega)(\rho-\sigma)\|_{1}\\
&=\frac{1}{1+\delta}\|\Psi(\rho-\sigma)\|_{1}\\
&\leq \frac{1}{1+\delta}\|\rho-\sigma\|_{1}.
\end{align*}
This shows that $\Ep^{\omega(\Ep)}$ is strictly contractive.



\end{proof}
Recall that for a  primitive quantum channel $$\Ep:\Md\rightarrow\Md$$ with $n$ linearly independent Kraus operators $\{a_1,\cdots, a_n\}$, $i(\Ep)$ is defined as follows:
\[i(\Ep)=\min\{k\in \mathbb{N}: \rm{Span}\{a_{i_1}\cdots a_{i_k}\}=\Md\}.\]
\begin{corollary}
For a primitive quantum channel $\Ep$, $\Ep^{i(\Ep)}$ is strictly contractive. 
\end{corollary}
\begin{proof}
As was shown in \cite{Wielandt}, we have the inequality 
$\omega(\Ep)\leq i(\Ep)$. Now if $i(\Ep)=\omega(\Ep)+m$, for some integer $m\geq 0$, then we have $\Ep^{i(\Ep)}=\Ep^m\circ\Ep^{\omega(\Ep)}$. By Theorem \ref{thm-strictly contrctv} we know that $\Ep^{\omega(\Ep)}$ is strictly contractive. It is easy to see that composition of any channel with a strictly contractive channel is also strictly contractive. Hence the result. 
\end{proof}
\section{Summary}
The main achievement of this article is twofold. Firstly, It improves the quantum Wielandt bound for a specific class of maps and this bound depends only on the dimension of the domain and not the maps themselves. Secondly, the maps considered here are positive maps as opposed to completely positive maps, hence working with weaker assumptions. The main  methodology used to achieve this new bound is the spectral behaviour of positive linear maps, specifically the multiplicative properties of positive maps and exploiting some related results in \cite{miza},\cite{Sam-Miza}.   
\section{Acknowledgements}
MR is supported by a Postdoctoral fellowship at the Department of Pure Mathematics, University of Waterloo. The author would like to thank Professor Vern Paulsen and Sam Jaques for many insightful discussion. We thank Prof. Mateusz Micha{\l}ek and Prof. Yaroslav Shitov for pointing out their recent work in this topic. Also we thank Micha{\l} Bia{\l}o\' nczyk for helpful discussion on this article. The author would also like to thank the anonymous referees for a lot of helpful suggestions especially related to the proof of the Theorem \ref{thm-strictly contrctv}. 
\bibliography{wielandt}

\providecommand{\bysame}{\leavevmode\hbox to3em{\hrulefill}\thinspace}
\providecommand{\MR}{\relax\ifhmode\unskip\space\fi MR }
\providecommand{\MRhref}[2]{%
  \href{http://www.ams.org/mathscinet-getitem?mr=#1}{#2}
}
\providecommand{\href}[2]{#2}
\begin{thebibliography}{10}

\bibitem{ergodic}
D.~Burgarth, G.~Chiribella, V.~Giovannetti, P.~Perinotti, and K.~Yuasa,
  \emph{Ergodic and mixing quantum channels in finite dimensions}, New J. Phys.
  \textbf{15} (2013), no.~July, 073045, 33. \MR{3094122}

\bibitem{choi1980b}
Man~Duen Choi, \emph{Some assorted inequalities for positive linear maps on
  {$C^{\ast} $}-algebras}, J. Operator Theory \textbf{4} (1980), no.~2,
  271--285. \MR{595415 (82c:46073)}

\bibitem{evans-krohn}
David~E. Evans and Raphael H{\o}egh-Krohn, \emph{Spectral properties of
  positive maps on {$C\sp*$}-algebras}, J. London Math. Soc. (2) \textbf{17}
  (1978), no.~2, 345--355. \MR{0482240}

\bibitem{farenick}
D.~R. Farenick, \emph{Irreducible positive linear maps on operator algebras},
  Proc. Amer. Math. Soc. \textbf{124} (1996), no.~11, 3381--3390. \MR{1340385}

\bibitem{Doug-Miza}
Douglas Farenick and Mizanur Rahaman, \emph{Bures contractive channels on
  operator algebras}, New York J. Math. \textbf{23} (2017), 1369--1393.
  \MR{3723514}

\bibitem{NP-hard}
S.~{Gaubert} and Z.~{Qu}, \emph{{Checking the strict positivity of Kraus maps
  is NP-hard}}, ArXiv e-prints (2014).

\bibitem{gurvits}
Leonid Gurvits, \emph{Classical complexity and quantum entanglement}, J.
  Comput. System Sci. \textbf{69} (2004), no.~3, 448--484. \MR{2087945}

\bibitem{math-language-qit}
Teiko Heinosaari and M\'ario Ziman, \emph{The mathematical language of quantum
  theory}, Cambridge University Press, Cambridge, 2012, From uncertainty to
  entanglement. \MR{2896337}

\bibitem{H-J}
Roger~A. Horn and Charles~R. Johnson, \emph{Matrix analysis}, Cambridge
  University Press, Cambridge, 1990, Corrected reprint of the 1985 original.
  \MR{1084815}

\bibitem{entng-brkng}
Michael Horodecki, Peter~W. Shor, and Mary~Beth Ruskai, \emph{Entanglement
  breaking channels}, Rev. Math. Phys. \textbf{15} (2003), no.~6, 629--641.
  \MR{2001114}

\bibitem{idel}
Martin Idel, \emph{On the structure of positive maps}, Master's thesis,
  Technische Universitat Munchen, 2013.

\bibitem{Sam-Miza}
Samuel Jaques and Mizanur Rahaman, \emph{Spectral properties of tensor products
  of channels}, J. Math. Anal. Appl. \textbf{465} (2018), no.~2, 1134--1158.
  \MR{3809349}

\bibitem{knill-lafl}
Emanuel Knill and Raymond Laflamme, \emph{Theory of quantum error-correcting
  codes}, Phys. Rev. A \textbf{55} (1997), 900--911.

\bibitem{KYE-96}
Seung-Hyeok Kye, \emph{Facial structures for the positive linear maps between
  matrix algebras}, Canad. Math. Bull. \textbf{39} (1996), no.~1, 74--82.
  \MR{1382493}

\bibitem{Mat-Shitov}
Mateusz Micha{\l}ek and Yaroslav Shitov, \emph{Quantum version of {W}ielandt's
  inequality revisited}, IEEE Trans. Inform. Theory \textbf{65} (2019), no.~8,
  5239--5242. \MR{3988551}

\bibitem{M-S-D-W}
Alexander M\"uller-Hermes, Daniel Stilck~Fran\c{c}a, and Michael~M. Wolf,
  \emph{Entropy production of doubly stochastic quantum channels}, J. Math.
  Phys. \textbf{57} (2016), no.~2, 022203, 22. \MR{3457927}

\bibitem{paulsen}
Vern Paulsen, \emph{Completely bounded maps and operator algebras}, Cambridge
  Studies in Advanced Mathematics, vol.~78, Cambridge University Press,
  Cambridge, 2002. \MR{1976867}

\bibitem{Garcia-etall}
D.~Perez-Garcia, F.~Verstraete, M.~M. Wolf, and J.~I. Cirac, \emph{Matrix
  product state representations}, Quantum Info. Comput. \textbf{7} (2007),
  no.~5, 401--430.

\bibitem{Raginsky}
Maxim Raginsky, \emph{Strictly contractive quantum channels and physically
  realizable quantum computers}, Phys. Rev. A \textbf{65} (2002), 032306.

\bibitem{miza}
Mizanur Rahaman, \emph{Multiplicative properties of quantum channels}, Journal
  of Physics A: Mathematical and Theoretical \textbf{50} (2017), no.~34,
  345302.

\bibitem{R-J-P}
Mizanur Rahaman, Samuel Jaques, and Vern~I. Paulsen, \emph{Eventually
  entanglement breaking maps}, J. Math. Phys. \textbf{59} (2018), no.~6,
  062201, 11. \MR{3810794}

\bibitem{Wielandt}
Mikel Sanz, David P\'erez-Garc\'ia, Michael~M. Wolf, and Juan~I. Cirac, \emph{A
  quantum version of {W}ielandt's inequality}, IEEE Trans. Inform. Theory
  \textbf{56} (2010), no.~9, 4668--4673. \MR{2807352}

\bibitem{seneta}
E.~Seneta, \emph{Non-negative matrices and {M}arkov chains}, Springer Series in
  Statistics, Springer, New York, 2006, Revised reprint of the second (1981)
  edition [Springer-Verlag, New York; MR0719544]. \MR{2209438}

\bibitem{Shir-Tat}
M.~E. Shirokov and Tatiana Shulman, \emph{On superactivation of zero-error
  capacities and reversibility of a quantum channel}, Comm. Math. Phys.
  \textbf{335} (2015), no.~3, 1159--1179. \MR{3320308}

\bibitem{stormer2007}
Erling St{\o}rmer, \emph{Multiplicative properties of positive maps}, Math.
  Scand. \textbf{100} (2007), no.~1, 184--192. \MR{2331197}

\bibitem{stormer}
\bysame, \emph{Positive linear maps of operator algebras}, Springer Monographs
  in Mathematics, Springer, Heidelberg, 2013. \MR{3012443}

\bibitem{Watrous-book}
John Watrous, \emph{The theory of quantum information}, Cambridge University
  Press, 2018.

\bibitem{Wielandt-original}
Helmut Wielandt, \emph{Unzerlegbare, nicht negative {M}atrizen}, Math. Z.
  \textbf{52} (1950), 642--648. \MR{0035265}

\bibitem{wolf}
M.~M. Wolf, \emph{Quantum channels and operations- guided tour}, Online Lecture
  Notes, 2012.

\end{thebibliography}
\bibliographystyle{amsplain}
E-mail address: mizanur.rahaman@uwaterloo.ca 

\end{document}